\documentclass[11pt]{article}

\newcommand{\tr}{^{\text{T}}}
\newtheorem{thm}{Theorem}

\newtheorem{lem}{Lemma}
\newenvironment{proof}{\paragraph{Proof:}}{\hfill\null}

\usepackage{amssymb}
\usepackage{amsmath}
\usepackage{multirow}
\usepackage{a4wide}

\def\R{\in \mathbb{R}}

\def\cov{\text{Cov}}

\def\lv{\left\Vert}
\def\rv{\right\Vert}

\def\E{\mathbb{E}}
\usepackage{graphicx,times,amsmath}

\begin{document}

\title{A PRESS statistic for two-block partial least squares regression}
\author{Brian McWilliams and Giovanni Montana\thanks{Email: {\tt g.montana@imperial.ac.uk}} \\Statistics Section, Department of Mathematics, Imperial College London, London, UK}

\maketitle

\begin{abstract}
Predictive modelling of multivariate data where both the covariates and responses are high-dimensional is becoming an increasingly popular task in many data mining applications. Partial Least Squares (PLS) regression often turns out to be a useful model in these situations since it performs dimensionality reduction by assuming the existence of a small number of latent factors that may  explain the linear dependence between input and output. In practice, the number of latent factors to be retained, which controls the complexity of the model and its predictive ability, has to be carefully selected. Typically this is done by cross validating a performance measure, such as the predictive error. Although cross validation works well in many practical settings, it can be computationally expensive. Various extensions to PLS have also been proposed for regularising the PLS solution and performing simultaneous dimensionality reduction and variable selection, but these come at the expense of additional complexity parameters that also need to be tuned by cross-validation. In this paper we derive a computationally efficient alternative to leave-one-out cross validation (LOOCV), a predicted sum of squares (PRESS) statistic for two-block PLS. We show that the PRESS is nearly identical to LOOCV but has the computational expense of only a single PLS model fit. Examples of the PRESS for selecting the number of latent factors and regularisation parameters are provided. 
\end{abstract}

\section{Introduction}
In this work we consider regression settings characterised by an $X\in \mathbb{R}^{n\times p}$ matrix of $p$ covariates and an $Y\in \mathbb{R}^{n\times q}$ matrix of $q$ responses, both observed on $n$ objects. Assuming centred data, the standard regression approach consists of fitting a Multivariate Linear Regression (MLR) model, that is $Y=X\beta + \epsilon$, where $\beta = (X\tr X)^{-1}X\tr Y \in \mathbb{R}^{p\times q}$ is a matrix of regression coefficients and $\epsilon \in \mathbb{R}^{n\times q}$ is the matrix of uncorrelated, mean-zero errors. In many situations where the dimensionality of the covariates is very high, the problem of multicollinearity prevents $X\tr X$ from being invertible. Furthermore, it can be shown that the columns, $[\beta_{1},...,\beta_{q}]$, of the matrix of regression coefficients $\beta$ are in fact the coefficients of the regression of $X$ on the individual response variables, $[y_{1},...,y_{q}]$, respectively (see, for instance, \cite{Izenman2008}). This implies that the least squares estimate for $\beta$ is equivalent to performing $q$ separate multiple regressions. Therefore, the MLR solution contains no information about the correlation between variables in the response. 

A common solution to the problems introduced by multicollinearity and multivariate responses involves imposing constraints on $\beta$ which improves prediction performance by effectively performing dimensionality reduction. Two-block Partial Least Squares (PLS) is a technique which performs simultaneous dimensionality reduction and regression by identifying two sets of latent factors underlying the data which explain the most covariance between covariates and response where the two blocks refer to the setting where both the covariates and the response are multivariate. These latent factors are then used for prediction rather than the full data matrices $X$ and $Y$. As such, PLS is used in many situations where the number of variables is very large. Two-block PLS is described in detail in section \ref{subs-pls}.

There are several modern applications for which  two-block PLS regression is particularly suitable. PLS is widely used in the field of chemometrics where there are many problems which involve predicting the properties of chemicals based on their composition where the latent factors capture the underlying chemical processes \cite{Wold2001}. In computational finance, PLS has been used to identify the latent factors which explain the movement of different markets in order to predict stock market returns \cite{Cengiz2010}. 

Closely related is the field of multi-task learning in which many tasks with univariate output are learned simultaneously. Using information from multiple tasks improves the performance and so can be treated as one large problem with a multivariate response (see, for example \cite{Argyriou2008}). For example, there are similarities in the relevance of a particular web page to a given query depending on geographical location and so learning the predictive relationship between query terms and web pages across multiple locations helps improve the relevance of the returned pages \cite{Chapelle2010}. In applications involving web search data, $n$ can be extremely large, in the order of hundreds of thousands.

%In all these situations there are often many correlated, noisy and irrelevant variables and as such usual methods such as ordinary least squares (OLS) regression are unable to capture the true structure of the predictive relationship and can sometimes fail outright. 

%On the other hand, there are applications  where $n$ can be extremely large. For instance, to web search data such as the  ranking problem which aims to match the most relevant web pages to a particular query involve scenarios where the number of observations may be in the order of hundreds of thousands \cite{Chapelle2010}. {\bf HOW DOES THIS FIT THE REGRESSION SETTING?} Furthermore, many problems in market research and business analytics aim to predict customer behaviour based on a few variables and many observations \cite{EspositoVinzi2010}. {\bf IS THIS SUITABLE?}

%Recently, PLS has been gaining popularity as a technique for integrative analysis of DNA copy number and gene expression \cite{Witten2009a, Parkhomenko2009} and other genomic data \cite{Cao2008,Chun2010}. 
One critical issues that arises when using PLS regression to attack such problems relates to the selection of the number of latent factors, $R$, to include in the model. The choice of $R$ is extremely important: if $R$ is too small, important features in the data may not be captured; if $R$ is too large the model will overfit. Furthermore, the number of latent factors can be important in interpreting the results. Recently, several extensions to PLS have been proposed to improve prediction in several settings. When there are many irrelevant variables which do not contribute to the relationship between $X$ and $Y$ (i.e. the underlying model of the data is sparse), we can regularize the latent factors to remove the contribution of noise variables from the regression coefficients. However, better prediction performance comes at the cost of introducing more important parameters such as the degree of sparsity which must be tuned. The success of these and the many other extensions to PLS depends on the careful tuning of these parameters.
 
For decades, performing $K$-fold cross validation (CV) on the prediction error has been a popular tool for model selection in linear models \cite{Stone1974a,Wahba1977}. In PLS regression, model selection has also been commonly performed using Leave-one-out cross validation (LOOCV) \cite{Cao2008} and $K$-fold CV \cite{Chun2010}. The CV procedure involves repeated fitting of the model using subsets of the data in the following way. The data is split into $K$ equal sized groups and the parameters are estimated using all but the $k^{th}$ group of observations which is used for testing. The $K$ groups are each left out in turn and the $K$-fold cross validated error is given by
\[
E_{cv} = \frac{1}{K} \sum_{k=1}^{K}(y_{k}-X_{k}\beta(k))^{2} \text{ ,}
\]
where the subscript $k$ denotes only the observations in the $k^{th}$ group are used whereas $(k)$ denotes the estimate of $\beta$ obtained by leaving out the $k^{th}$ group.

The choice of $K$ is important as it has implications on both the accuracy of the model selection as well as its computational cost. %\cite{Stone1974}.
When $K=n$ we obtain leave-one-out cross validation (LOOCV) where the parameters are estimated using all but one observation and evaluated on the remaining observation. LOOCV is a popular choice for model selection as it makes most efficient use of the data available for training. LOOCV also has the property that as the number of samples increases, it is asymptotically equivalent to the Akaikie Information Criterion (AIC) which is a commonly used for model selection in a variety of applications \cite{Stone1977}. However LOOCV can be extremely computationally expensive if $n$ or $p$ is large. %It is known that the AIC has low bias but high variance so in this limit, LOOCV will tend to select more complicated models \cite{Shao1993}. However this property suggests that in situations where $n$ is small, LOOCV is a more suitable method than $K-$fold CV which suffers from high bias.

The use of these techniques are subject to the same problems as in OLS, however, since PLS is typically used in situations where the data is very high dimensional and the sample size is small, the problems can be amplified. When $n$ is small compared to $p$ and $q$, constructing $K$ cross-validation sets is wasteful and the resulting model selection becomes less accurate. Similarly, although LOOCV selects the true model with better accuracy under these conditions, the computational cost becomes prohibitive when $p$, $q$ or $n$ are large since the complexity of PLS is of order $O(n(p^2q + q^2p))$ \cite{Golub1996}. Problems of performing model selection in PLS are discussed by \cite{McWilliams2010}. 

It has long been known that for ordinary least squares, the LOOCV of the prediction error has an exact analytic form known as the Predicted Sum of Squares (PRESS) statistic \cite{Belsley1980}. Using this formulation it is possible to rewrite the LOOCV for each observation as a function of only the residual error and other terms which do not require explicit partitioning of the data. We briefly review this result in Section \ref{subs-ols}. However, such a method for computing PRESS efficiently for two-block PLS has not been developed in the literature. In this work we derive a analytic form of PRESS for two-block PLS regression based on the same techniques as PRESS for OLS which we present in section \ref{subs:plspress}. In section \ref{sec:proof} we show that, under mild assumptions, the PRESS is equivalent to LOOCV up to an approximation error of order $O (\sqrt{\frac{\log n}{n}})$. In section \ref{sec:applications} we illustrate how the PLS PRESS can be used for efficient model selection with an application to Sparse PLS where the PRESS is used to select the number of latent factors $R$ and the regularization parameter controlling the degree of sparsity in the regression coefficients. Finally, we report on experiments performed using data simulated under the sparse PLS model and show that the PRESS performs almost exactly as  LOOCV but at a much smaller computational cost.

% -- whereas the PRESS is linear in $n$, $p$ and $q$, LOOCV is quadratic in these quantities which results in huge computational savings when performing model selection in datasets involving any combination of large $n$, $p$ or $q$. 

%We then report on experiments performed on data simulated under the PLS model which show that PRESS achieves exactly the same sensitivity as LOOCV in a wide variety of simulation settings where $n$ is both smaller and larger than both $p$ and $q$. We also show that the PRESS is linear in $n$, $p$ and $q$ whereas LOOCV is quadratic in these quantities which results in huge computational savings when performing model selection in datasets involving any combination of large $n$, $p$ or $q$. 

\section{PLS Regression} \label{subs-pls} 

PLS is a method for performing simultaneous dimensionality reduction and regression by identifying a few latent factors underlying the data which best explain the covariance between the covariates and the response. Regression is then performed using these lower dimensional latent factors rather than the original data in order to obtain better prediction performance since only the features in the data which are important for prediction are used.

The two-block PLS model assumes $X$ and $Y$ are generated by a small number $R$, of latent factors \cite{Rosipal2006, Wegelin2000, McWilliams2010}
\begin{eqnarray}
X = TP\tr + E_x \nonumber ~,~~ Y = SQ\tr + E_y, 
\label{pls_model}
\end{eqnarray} 
where the columns of $T$ and $S \R ^{n\times R}$ are the $R$ latent factors of $X$ and $Y$ and $P\R^{p \times R}$ and $Q \R^{q \times R}$ are the factor loadings. $E_x\R ^{n \times p}$ and $E_y \R ^{n \times q}$ are the matrices of residuals. The latent factors are linear combinations of all the variables and are found so that
\[
 \cov (T,S)^{2} = \max_{U,V} \cov (XU,YV)^{2},
\]
where $U \R ^{p \times R}$ and $V \R ^{ q \times R}$ are found by computing the singular value decomposition (SVD) of $M=X\tr Y $ as $M=UGV\tr$ where $U$ and $V$ are orthogonal matrices and $G$ is a diagonal matrix of singular values. There exists an ``inner-model'' relationship between the latent factors
\begin{equation}
S = TD + H,
\label{eq:inner}
\end{equation}
where $D\R ^{R\times R}$ is a diagonal matrix and $H$ is a matrix of residuals so that the relationship between $X$ and $Y$ can be written as a regression problem
\begin{eqnarray}
Y & = & T D Q\tr + (HQ \tr + E_y) \nonumber \\
  & = & X\beta + E_y^{*}, \nonumber
\end{eqnarray}
where $\beta= UDQ\tr$.
Since $U$ and $V$ are computed using the first $R$ singular vectors of the SVD of $M$, the $R$ PLS directions are orthogonal and so the latent factors can be estimated independently.

\section{PRESS statistic}

\subsection{The PRESS statistic for OLS regression} \label{subs-ols}
In the least squares framework, leave one out cross validation is written as
\[
E_{LOOCV} = \frac{1}{n} \sum_{i=1}^{n} \left( y_i - x_i\beta^{OLS}(i) \right)^{2}, 
\]
where $\beta^{OLS}$ are the OLS regression coefficients. The subscript $i$ denotes the $i^{th}$ observation whereas $(i)$ denotes the estimate of $\beta$ using the observations $(1,...,i-1, i+1,...,n)$. Given a solution $\beta(j)$, $\beta(i)$, $j\neq i$ is estimated by adding the $j^{th}$ observation and removing the $i^{th}$ observation. Estimating $\beta$ requires computing the inverse covariance matrix, $P=(X\tr X)^{-1}$ which is computationally expensive. However, since each $\beta(i)$ is different from $\beta$ by only one sample, we can easily compute $P(i)=(X(i)\tr X(i))^{-1}$ from $P$ using the Morrison-Sherman-Woodbury theorem %$P(i)$ can be efficiently updated to remove a single observation in the following way 
without the need to perform another matrix inversion \cite{Belsley1980}:
\begin{align*}
\left(X(i) \tr X(i)\right)^{-1} & = \left(X\tr X - x_i\tr x_i \right)^{-1} \\ 
 									              &= P + \frac{P x_i x_i \tr P}{{1 - h_i}},
\end{align*}
where $h_i = x_i\tr P x_i$. This allows the leave-one-out estimate, $\hat{\beta}(i)$ to be written as a function of $\hat{\beta}$ in the following way, without the need to explicitly remove any observations
\begin{align*}
\hat{\beta}(i)  & =  \left( X(i)\tr X(i) \right)^{-1}X(i)\tr y(i) \\ 
		 		  & =  \left( \hat{\beta} - \frac{(y_i - x_i \beta) P x_i}{1 - h_i} \right).
\end{align*}
Finally, the $i^{th}$ LOOCV residual can simply be written as a function of the $i^{th}$ residual error and does not require any explicit permutation of the data, as follows
\begin{align}
%e(i) &= y_i - x_i\tr \hat{\beta}(i) \nonumber \\
e(i)		 &= \frac{e_i}{1-h_i}. \nonumber
\end{align}

In the next section we derive a similar formulation for the PRESS for PLS.

\subsection{A PRESS statistic for two-block PLS}

\label{subs:plspress}

As described in section \ref{subs-pls}, estimating the PLS regression coefficients involves first estimating the inner relationship between latent factors and then the outer relationship between $S$ and $Y$. Both of these steps are the result of lower dimensional least squares regression problems. Using the same techniques as in section \ref{subs-ols}, we can derive a similar expression for the PLS PRESS statistic in two steps.
However, in PLS the latent factors are estimated by projecting the data onto the subspace spanned by $U$ and $V$ respectively which are found using the SVD of $M=UGV\tr$. 

Our developments rely on a mild assumption: provided the sample size $n$ is sufficiently large, any estimate of the latent factors obtained using $n-1$ samples is {\it close enough} to the estimate obtained using $n$ data points. In terms of the PLS model, this assumption relates to the difference between $g_r(M_n)$, the $r^{th}$ singular value of $M_{n}$ (which has been estimated using $n$ rows of $X$ and $Y$) and $g_r(M_{n-1})$, estimated using $n-1$ rows.

Formally, we assume that the following inequality 
\begin{equation}
\left| g_r(M_n) - g_r(M_{n-1})  \right| \leq \epsilon  \label{assumption}
\end{equation}
holds for $1\leq r \leq R$ where the approximation error, $\epsilon$ is arbitrarily small. Since the rank $r$ approximation error of the SVD is given by $g_{r+1}$, if the difference  between the pairs of singular values is small, it implies the difference between the corresponding pairs of singular vectors is also small. In other words, within the LOOCV iterations, it is not necessary to recompute the SVD of $X(i)\tr Y(i)$. We show that this assumption holds in Section \ref{sec:proof}. 

This assumption \eqref{assumption} implies that the $i^{th}$ PLS prediction error is $e(i)=y_i - x_i U D(i) Q(i)\tr$. Since the PLS inner model coefficient, $D$ in Eq \eqref{eq:inner} is estimated using least squares, we can derive an efficient update formula for $D(i)$ as a function of $D$ using the Morrison-Sherman-Woodbury theorem 
\begin{align*}
D(i) & = \left( T(i)\tr T(i) \right)^{-1}T(i)\tr S(i) \nonumber \\
		 %& = & \left( T\tr T - t_i \tr t_i \right)^{-1}(T\tr S - t_i \tr s_i) \nonumber \\
		 & = \left( D - \frac{(s_i - t_i D) P_t t_i}{1 - h_t} \right),
\end{align*}
where $P_t = (T\tr T)^{-1}$ is the inverse covariance matrix of the $X-$latent factors and $h_{t,i} = t_i \tr P_t t_i$ is analogous to the OLS hat-matrix. Similarly, the $Y-$loading $Q$ is also obtained using least squares and we can derive a similar update formula as follows
$$
%Q(i) & = & \left( S(i)\tr S(i) \right)^{-1}S(i)\tr Y(i) \nonumber \\
Q(i)      = \left( Q - \frac{(y_i - s_i Q) P_s s_i}{1 - h_s} \right),
$$
where $P_s = (S\tr S)^{-1}$ and $h_{s,i} = s_i \tr P_s s_i$. Under the assumption \eqref{assumption}, $U$ and $V$ are fixed and so these recursive relationships allow us to formulate the $i^{th}$ PLS prediction error, $e(i)$ in terms of the $i^{th}$ PLS regression residual error $e_i$ in the following way for one latent factor:
\begin{eqnarray}
e^{(r)}(i) &=&  y_i - x_i \sum_{r=1}^{R} U^{(r)} D^{(r)}(i) Q^{(r)}(i) \nonumber \\
		 %&=& \sum_{r=1}^{R} e^{(r)}_i \left[ 1 - \frac{ah_t}{1-h_t} - \frac{b+h_s}{1-h_s} \right. \nonumber \\		 
		 %&& - \left. \frac{bh_t-ah_th_s}{(1-h_t)(1-h_s)} \right] \nonumber \\
		 &=& \sum_{r=1}^{R}  e^{(r)}_i \left[ \frac{1+a}{(1-h_s)} + \frac{-a-b}{(1-h_t)(1-h_s)} \right], \nonumber 
\end{eqnarray}
where the following identities from the PLS model have been used: $s_i Q\tr = y_i - E_{y,i}$ $t_i=x_iU$, $s_i=x_iUD$ and $\beta = x_iUDQ\tr$.
where: $a=1-E_y/e_i$ and $b=y_i H_i P_s s_i/e_i$.

Since each of the PLS latent factors is estimated independently, for $R>1$ the overall leave one out error is simply the sum of the individual errors
\begin{eqnarray}
%e(i) &=&  y_i - x_i \sum_{r=1}^{R} U^{(r)} D^{(r)}(i) Q^{(r)}(i) \nonumber \\
	e(i)	 &=& \sum_{r=1}^{R} e^{(r)}(i) - y_i(R-1). \nonumber
\end{eqnarray}
Finally, the overall PRESS statistic is computed as
\begin{eqnarray}
E_{PRESS} = \frac{1}{n}\sum_{i=1}^{n} \left\Vert e(i) \right\Vert_2^{2},
\end{eqnarray}
where $\lv \cdot \rv^2_2$ is the squared Euclidean norm.
%where each $e(i)$ does not contain any quantities which must be computed by leaving out one observation.

\subsection{A bound on the approximation error}
\label{sec:proof}
The assumption, \eqref{assumption} is key to developing an efficient version of the PRESS for PLS. It ensures that the SVD need only be computed once with all the data instead of $n$ times in a leave-one-out fashion which results in a computational saving of $O(n(p^2q + q^2p))$ \cite{Golub1996} which becomes prohibitively expensive when both $p$ and $q$ are large. From a conceptual point of view, this assumption implicitly states that underlying the data there are true latent factors which contribute most of the covariance between $X$ and $Y$. Therefore, removing a single observation should not greatly affect the estimation of the latent factors. In this section we  formalise this assumption by introducing a theorem which places an upper bound on error $\epsilon$. 

In presenting the theorem, we first rely on two Lemmas. The first result, from details an upper bound on the difference between the expected value of a covariance matrix, of a random vector $x$ and its sample covariance matrix using $n$ samples.
\begin{lem}[adapted from  \cite{Rudelson2007}]
\label{thm:rud}
Let $x$ be a random vector in $\mathbb{R}^p$. Assume for normalization that $\Vert \E x\tr x\Vert_2 \leq 1$. Let $x_1,...,x_n$ be independent realizations of $x$.
Let
\[
m_n=C\sqrt{\frac{\log n}{n}}A
\] 
where $C$ is a constant and $A=\lv x \rv _2$. If $m<1$ then
\[
\E\left\Vert \frac{1}{n} \sum_{i=1}^{n} x_i\tr x_i - \E x\tr x \right\Vert_2 \leq m_n
\]
\end{lem}
\bigskip
The second Lemma is a result from matrix perturbation theory which details an upper bound on the maximum difference between the singular values of a matrix $M$ and a perturbation matrix $M+E$.
\begin{lem}[adapted from \cite{Sewart1990}]
For $M,E \R^{n\times p}$,
\[
\max_i\vert g_i(M+E) - g_i(M)\vert \leq \Vert E \Vert_2
\]
\label{lem:sv}
\end{lem}
We are now able to state our result. 
\begin{thm}
\label{thm:svd}
Let $M_n$ be the sample covariance matrix of $X_n\R ^{n\times p}$ and $M_{n-1}$ be the sample covariance matrix of $X_{n-1}\R ^{(n-1)\times p}$. $g_{r}(M)$ is the $r^{th}$ ordered singular value of $M$ then
\[
\max_i\left| g_i(M_n) - g_i(M_{n-1})  \right| \leq m_{n-1}
\]
\end{thm}

\begin{proof}
To prove this theorem, we first establish a bound on the error between the sample covariance matrices $M_n$ and $M_{n-1}$ by applying Lemma \ref{thm:rud} with $n$ and $n-1$ samples to obtain the following inequalities\begin{eqnarray}
\left\Vert M - M_n \right\Vert_2 &\leq & m_{n}\label{eq:n}\\
\left\Vert M - M_{n-1} \right\Vert_2 &\leq & m_{n-1} \label{eq:nmin}
\end{eqnarray}
subtracting Eq \eqref{eq:n} from Eq \eqref{eq:nmin} and applying Minkowski's inequality we arrive at an expression for the difference between terms $M_n$ and $M_{n-1}$ as follows
\[
\left\Vert M - M_{n-1} \right\Vert_2 - \left\Vert M_n - M_{n-1} \right\Vert_2 \leq m_n - m_{n-1}
\]
%\[
%\left\Vert M_n - M_{n-1} \right\Vert_2 \leq  \left\Vert M - M_{n-1} \right\Vert_2 - \left\Vert M_n - M_{n-1} \right\Vert_2
%\]
\begin{equation}
\left\Vert M_n - M_{n-1} \right\Vert_2 \leq  m_{n-1}
\label{eq:amin}
\end{equation}
We now relate this result to the difference between computing the SVD of $M_n$ and the SVD of $M_{n-1}$ by recognizing that $M_n$ is obtained as a result of perturbing $M_{n-1}$ by $M_1=x_1\tr x_1$ where $x_1\R^{1\times p}$ is the single observation missing from $M_n$. Using Lemma \ref{lem:sv} and Eq \eqref{eq:amin} we obtain
\begin{eqnarray}
\max_i\vert g_i(M_{n-1}+M_1) - g_i(M_{n-1})\vert & \leq & \Vert M_{n}-M_{n-1} \Vert_2 \nonumber \\
& \leq &   m_{n-1}
\end{eqnarray}
which proves the theorem.
\end{proof}%We can apply this inequality using $M_n=M_{n-1}+M_1$, where $M_1=x_1'x_1$ where $x_1\R^{1\times p}$ is the single observation missing from $M_n$. Using $E=M_{n}-M_{n-1}$ and Eq \eqref{eq:amin} we obtain
This theorem details an upper bound on the maximum difference between pairs of ordered singular values of the covariance matrix of $M$ estimated with all $n$ observations and the covariance matrix estimated with $n-1$ observations. Since $A$ and the constant $C$ do not depend on $n$ and so are the same for $m_n$ and $m_{n-1}$. Therefore, the value of the error term defined by the bound decreases as $O (\sqrt{\frac{\log n}{n}})$.
\bigskip
%This theorem details an upper bound on the maximum difference between pairs of ordered singular values of the covariance matrix of $X$ estimated with all $n$ observations and the covariance matrix estimated with $n-1$ observations. The value of the error term defined by the bound is dependent on the size of $n$ and decreases as $n$ increases. 

\section{Model selection in Sparse PLS}

\label{sec:applications}
%As mentioned, PLS is typically applied in situations where the data is high dimensional and highly correlated where there are fewer samples than variables. A class of regularized PLS regression methods have been developed to tackle the challenges posed by such situations. In all these cases, new parameters are introduced which control the degree of regularization, the careful tuning of which can be used to improve prediction and interpretibility of results. When the number of samples is extremely small, the PLS covariance matrix is often rank deficient and so $R$ latent factors may not be extracted accurately. In this case, the covariance matrix can be regularized using a ridge parameter which causes it to become full rank \cite{McWilliams2010}. 

Although the dimensionality reduction inherent in PLS is often able to extract important features from the data, in some situations where $p$ and $q$ are very large, there may be many irrelevant and noisy variables which do not contribute to the predictive relationship between $X$ and $Y$. Therefore, for the purposes of obtaining a good prediction and for interpreting the resulting regression coefficients, is it important to determine exactly which variables are the important ones and to construct a regression model using only those variables. Sparse PLS regression has found many applications in various areas, including genomics \cite{Cao2008}, computational finance \cite{McWilliams2010} and machine translation \cite{Hardoon2009}.

A Sparse PLS algorithm \cite{McWilliams2010} can be obtained by rewriting the SVD, $M=UGV\tr$ as a LASSO penalized regression problem
\begin{equation}
\min_{\tilde{u},\tilde{v}}\left\Vert M-\tilde{u}\tilde{v}\tr \right\Vert_2
^{2}+\gamma\left\Vert \tilde{u} \right\Vert_{1} ~~~ \text{s.t.}
~~\left\Vert \tilde{v} \right\Vert_2=1 \label{EQspls} \text{ ,}
\end{equation}
where $\tilde{u}$ and $\tilde{v}\in\mathbb{R}^{p\times1}$ are the estimates of
$u^{(1)}$ and $v^{(1)}$, the first left and right singular vectors respectively. As such, they are restricted to be vectors with unit
norm so that a unique solution may be obtained. The amount of sparsity in the solution is controlled by $\gamma$. If
$\gamma$ is large enough, it will force some variables to be
exactly zero. The problem of Eq. \eqref{EQspls} can be solved in an iterative
fashion by first setting  $\tilde{u}=u^{(1)}$ and
$\tilde{v}=v^{(1)}$ as before. The Lasso penalty can then be applied as a
component-wise soft thresholding operation on the elements of
$\tilde{u}$ (see, for instance, \cite{Friedman2007}). The sparse $\tilde{u}$ are
found by applying the threshold component-wise as follows:
\begin{eqnarray*}
\label{penalty}
\tilde{u}^{*} & = & \text{sgn}\left(H v\right)\left(\left|H
v\right|-\gamma\right)_{+}\\
\tilde{v}^{*} & = & H\tilde{u}^{*}/\left\Vert H\tilde{u}^{*}\right \Vert_2 \text{
.}
\end{eqnarray*} 

%The number of latent factors is commonly selected using LOOCV \cite{Wold1984} \cite{Cao2008} and $K$-fold CV \cite{Chun2010}. However, in situations where the data is extremely high dimensional and the sample size is small, $K$-fold CV no longer accurately selects the correct model and the computational cost of LOOCV becomes prohibitive. 
Typically, selecting the optimal sparsity parameter, $\lambda$ involves evaluating the CV error over a grid of points in $[\lambda_{min},\lambda_{max}]$. However, since the behaviour of the CV error as a function of $\lambda$ is not well understood this often requires specifying a large number of points which exacerbates the already large computational burden of performing cross validation. In such cases, model selection is often performed using problem-specific knowledge and intuition. %\cite{Parkhomenko2009}. 
Problems of performing model selection in sparse PLS are discussed by \cite{McWilliams2010}. Since the PLS PRESS is computationally efficient for large $n,p$ and $q$ we can evaluate the Sparse PLS model over a large grid of values of $\lambda$ and quickly compute the PRESS to determine the best cross-validated degree of sparsity. Some experimental results are presented in the following Section.

\section{Experiments}

\label{sec:data}
In this section we report on the performance of the PRESS in experiments on simulated data where the predictive relationship between the covariates and the response is dependent on latent factors according to the PLS model in Eq \eqref{pls_model}
where $E_x \R ^{n\times p}$ and $E_y \R ^{n\times q}$ are matrices of i.i.d random noise simulated from a normal distribution, $N(0,1)$. For a fixed value of $R$, $n$ and $p=q$ we simulate $R$ pairs of latent factors $t$ and $s$ of length $n$ from a bivariate normal distribution in descending order of covariance with means drawn from a uniform distribution. Secondly we simulate $R$ separate pairs of loading vectors $u$ and $v$ of length $p$ and $q$ respectively from a uniform distribution, $U(0,1)$. In order to ensure the contribution of each latent factor is orthogonal, we orthogonalise the vectors $U=[u_1,...,u_R]$ with respect to each other using the $QR$ decomposition, similarly for $V=[v_1,...,v_R]$. %Finally, we construct the data matrices $X$ and $Y$ according to the PLS model
%\begin{align}
%X = \sum_{r=1}^{R} t_r u_r + E_x  ~,~~ Y = \sum_{r=1}^{R} s_r v_r + E_y \nonumber
%\end{align}
%where $E_x \R ^{n\times p}$ and $E_y \R ^{n\times q}$ are matrices of i.i.d random noise simulated from $N(0,1)$. 

To test the performance of the PRESS for selecting the number of latent factors, we perform a Monte Carlo simulation where for each iteration we draw $R$ as an integer from $U(2,8)$ so that the true number of latent factors in the simulated data is constantly changing. We measure the sensitivity of model selection using the PRESS and LOOCV by comparing the number of latent factors which minimizes these quantities minimum with the true value of $R$.

To test the performance of the PRESS for sparse PLS, we use the same simulation setting except now we fix $R=1$ and induce sparsity into the latent factor loadings for $X$, $u_r$. Now $\lv u_r \rv _0=p/j$ which implies that only $p/j$ of the $p$ elements in $u_r$ are non-zero and thus contribute to the predictive relationship between $X$ and $Y$. By altering $j$, we change how many of the variables in $X$ are useful for predicting $Y$. 

We perform a Monte Carlo simulation whereby for each iteration we randomize the true number of important variables in $X$ by drawing $j$ from $U(1,2)$ so that up to half of the variables in $X$ may be unimportant. We evaluate sparse PLS over a grid of $100$ values of $\gamma$ which span the parameter space. We measure the sensitivity of model selection using the PRESS as compared to LOOCV by comparing the selected variables with the truly important variables. 

Table I reports on the ratio between the sensitivity achieved by the PRESS, $\pi_{PRESS}$ and the sensitivity of LOOCV, $\pi_{LOOCV}$ for both of these settings averaged over 200 trials for different values of $n$, $p$ and $q$. When selecting $R$, PRESS and LOOCV achieve almost exactly the same sensitivity for all values with only very small discrepancies when $n<p$. When selecting $\gamma$ the error between PRESS and LOOCV is noticeable when $p$ and $q$ are large compared with $n$. The standard error of the sensitivity ratio (in parenthesis) also increases when $p$ and $q$ are large.

\begin{table}[h]
%\caption{Comparison of sensitivity}
\begin{center}
\begin{tabular}{|r|r|l|l|l|l|}
\hline
& & \multicolumn{4}{|c|}{$\pi_{PRESS} / \pi_{LOOCV}$}  \\
\hline
\multicolumn{1}{|c|}{$n$}%{\raisebox{-1.50ex}[0cm][0cm]{\!number of variables\!}}
%& \multicolumn{1}{|c|}{Residual}% time}
& \multicolumn{1}{|c|}{$p,q$}
& \multicolumn{2}{|c|}{Selecting $R$}% time}
& \multicolumn{2}{|c|}{Selecting $\gamma$}% cost} 
\\ \hline
 \multirow{3}{*}{50}   & 100 &   $1.0$ & $(0.303)$ &  $0.92$ & $(0.228)$ \\ %\hline
   										 & 500 &   $0.98$ & $(0.390)$ & $0.73$ & $(0.344)$  \\ %\hline
    									 & 1000 &  $0.95$ & $(0.351)$ & $0.70$ & $(0.333) $ \\ \hline
 
 \multirow{3}{*}{100}  & 100 &   $1.0$ & $(0.141)$ &  $0.99$ & $(0.164)$ \\ %\hline
   										 & 500 &   $0.99$ & $(0.412)$ &  $0.91$ & $(0.190)$  \\ %\hline
    									 & 1000 &  $0.99$ & $(0.389)$ &  $0.90$ & $(0.203)$ \\ \hline
 
 \multirow{3}{*}{200}  & 100 &   $1.0$ & $(0.0)$ &  $1.0$ & $(0.071)$ \\ %\hline
   										 & 500 &   $1.0$ & $(0.095)$ &  $0.95$ & $(0.055)$  \\ %\hline
    									 & 1000 &  $1.0$ & $(0.332)$ &  $0.91$ & $(0.158)$ \\ \hline
    									 
 \multirow{3}{*}{300}  & 100 &   $1.0$ & $(0.0)$ &  $1.0$ & $(0.014)$ \\ %\hline
   										 & 500 &   $1.0$ & $(0.1)$ &  $1.0$ & $(0.055)$ \\ %\hline
    									 & 1000 &  $1.0$ & $(0.127)$ &  $0.96$ & $(0.11)$ \\ \hline   
 
 \multirow{3}{*}{500}  & 100 &   $1.0$ & $(0.0)$ &  $1.0$ & $(0.0)$ \\ %\hline
   										 & 500 &   $1.0$ & $(0.095)$ &  $1.0$ & $(0.0)$  \\ %\hline
    									 & 1000 &  $1.0$ & $(0.12)$ &  $1.0$ & $(0.017)$ \\ \hline    									  									 
\end{tabular}
\label{tab-sens}
\end{center}
\caption{Comparing the ratio of the sensitivity, $\pi$ (the proportion of times the correct model is chosen), $\pi_{PRESS}/\pi_{LOOCV}$ when selecting $R$ and $\gamma$ as a function of $n$, $p$ and $q$. The value in parenthesis is the Monte Carlo standard error. %For sufficiently large $n$, both methods achieve the same sensitivity. For smaller values of $n$ the performance of PRESS is only marginally worse
}
\end{table}

Figure \ref{fig_Vars} compares the computational time using PRESS and LOOCV for selecting $\gamma$ as a function of $n$ for different values of $p$ and $q$. We report on computational timings using a 2.0GHz Intel Core2Duo with 4GB of RAM. Relative timing is measured using the \texttt{tic, toc} function in Matlab v7.8.0. It can be seen that increase in computation time for LOOCV is quadratic as a function of $n$ and the number of variables. In comparison, the increase in computation time for PRESS is linear in these quantities and very small relative to LOOCV. Because of this, it becomes computationally prohibitive to perform LOOCV with a greater number of observations or variables than we have presented.

Figure \ref{fig_Obserr} reports on the approximation error between LOOCV and PRESS for $p=q=100$ as a function of $n$. As the number of samples increases, it can be seen that the error decreases as $O(\sqrt{\frac{\log(n)}{n}})$.

\begin{figure}[htp]
\centerline{\includegraphics[width=3.35in]{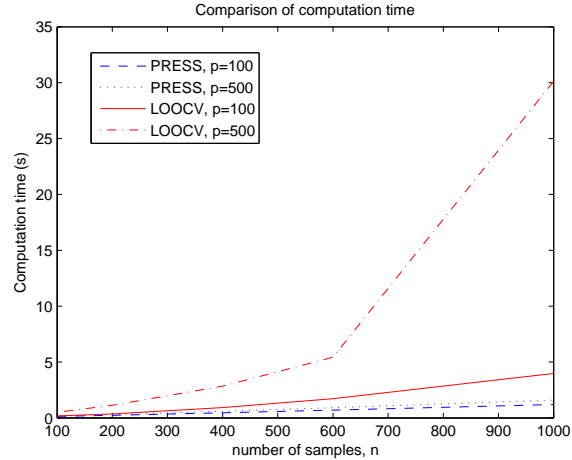}}
\caption{Comparison of computational timing between PRESS and LOOCV. For fixed values of $p$ and $q$, computational time (in seconds) for PRESS is linear in the number of observations, $n$  whereas for LOOCV, the computational time increases as a function of $n^2$. PRESS is also linear in the number of variables whereas LOOCV is quadratic in this quantity.}
\label{fig_Vars}
\end{figure}

\begin{figure}[htp]
\centerline{\includegraphics[width=3.35in]{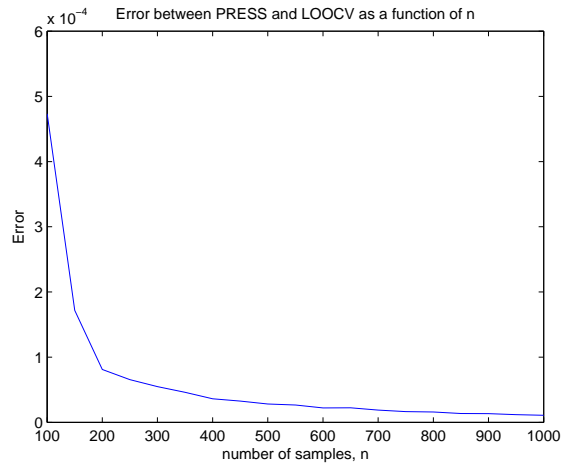}}
\caption{The approximation error between LOOCV and PRESS. For a fixed $p$ and $q=100$, the approximation is small. As the number of samples, $n$ increases, the error between LOOCV and PRESS decreases as $O(\sqrt{\frac{\log(n)}{n}})$. See table \ref{tab-sens} for further results.}
\label{fig_Obserr}
\end{figure}

%% REMOVE THIS FIGURE!!

In the simulations we have focussed on situations where the response is multivariate. However, the case where the response is univariate $(q=1)$ also commonly occurs. In this situation, the latent factor for $Y$ collapses to a scalar and so we would expect the error between LOOCV and PRESS to be smaller.

%We perform a Monte Carlo simulation whereby for each iteration we draw $R$ as an integer from $U(2,8)$ so that the true number of latent factors in the simulated data is constantly changing. We measure the sensitivity of model selection using the in sample residual as compared to LOOCV and PRESS by comparing the number of latent factors at the minimum with the true value of $R$. For $n=200$, we repeat this experiment for values of $p$ between $200$ and $2000$ where $q=p$. Table \ref{tab-sens} reports on the sensitivity averaged over 100 trials for each value of $p$ It can be seen that PRESS achieves close to maximum sensitivity in the case where there are far more variables than observations whereas the residual becomes less sensitive. PRESS and LOOCV achieve exactly the same sensitivity which confirms that our method for computing PRESS is equivalent to LOOCV.

%\begin{figure}[htp]
%\centerline{\includegraphics[width=3.35in]{figs/timeObs.eps}}
%\caption{For a fixed number of variables, computational time (in seconds) for PRESS is linear in the number of observations, $n$ whereas for LOOCV, the computational time increases as a function of $n^2$.}
%\label{fig_Obs}
%\end{figure}

%Figure \ref{fig_Obs} reports the computational time as a function of the number of observations which again shows that the complexity of PRESS is linear in $n$ whereas LOOCV again, quadratic. 

\section{Conclusions}
We have showed that in order to obtain good prediction in high dimensional settings using PLS, a computationally efficient and accurate model selection method is needed to tune the many crucial parameters. In this work we have derived an analytic form for the PRESS statistic for two-block PLS regression and we have proved that under mild assumptions, the PRESS is equivalent to LOOCV. %Simulation results show that the PRESS exhibits the same sensitivity as LOOCV at a far lower computational cost.

We have also presented an example application where the PRESS can be used to tune additional parameters which improve prediction and interpretability of results in the case of sparse PLS. We have showed through simulations that using the PRESS to tune the number of latent factors, $R$ and the sparsity parameter $\gamma$, performs almost identically to LOOCV which is the method most commonly used in the literature at a far lower computational cost. When the number of samples is large, LOOCV and PRESS perform identically.% Since sparse PLS is used in situations where the data is extremely high dimensional, the ability to perform LOOCV using the PRESS at a fraction of the computational cost is useful. 

Although we have showed that the analytic PRESS statistic for two-block PLS regression is an important contribution and can be easily applied to the many settings where parameters must be tuned accurately, there are still opportunities for further work. Another such possibility is to construct an objective function in terms of $\gamma$, the regularization parameter, so that the optimal degree of sparsity can be tuned automatically. 

\bibliographystyle{IEEEtran}
\bibliography{IEEEabrv,press}

\end{document}